\documentclass[11pt]{article} % Stat tech report
\newcommand{\doctype}{TECH}

\usepackage{fullpage}  % use paper more efficiently
\usepackage{graphicx,epsfig}  % need this to include pictures
\usepackage{amsmath}   % i like this for the added math macros, but
\usepackage{amsfonts,amsthm,amssymb}
\usepackage{epic,eepic,eepicemu}
\usepackage{epsf}
\usepackage{graphics}
\usepackage{psfrag}
\usepackage{epsfig}
\usepackage{ifthen}
\usepackage{bbm}
\usepackage{color}
\newcommand{\real}{\mathbb{R}}

\newcommand{\bdm}{\begin{displaymath}}
\newcommand{\edm}{\end{displaymath}}
\newcommand{\bea}{\begin{eqnarray*}}
\newcommand{\eea}{\end{eqnarray*}}
\newcommand{\bean}{\begin{eqnarray}}
\newcommand{\eean}{\end{eqnarray}}

\newcommand{\prob}{\mathbb{P}}
\newcommand{\expec}{\mathbb{E}}
\newcommand{\var}{\mathrm{Var}}

\newcommand{\numnode}{\ensuremath{m}}

\newcommand{\nodenum}{\numnode}
\newcommand{\defn}{\ensuremath{: \, = }}

\newcommand{\convdist}{\ensuremath{\stackrel{d}{\rightarrow}}}

\newcommand{\Exs}{\ensuremath{\mathbb{E}}}

\newcommand{\nodeSet}{\ensuremath{V}}

\newcommand{\vertex}{\ensuremath{v}}

\theoremstyle{plain}
\newtheorem{theorem}{Theorem}

\newtheorem{corollary}{Corollary}

\newtheorem{lemma}{Lemma}

\newcommand{\widgraph}[2]{\includegraphics[keepaspectratio,width=#1]{#2}}

\newcommand{\myparagraph}[1]{{\bf{#1}}}

\makeatletter \long\def\@makecaption#1#2{
        \vskip 0.8ex
        \setbox\@tempboxa\hbox{\small {\bf #1:} #2}
        \parindent 1.5em  %% How can we use the global value of this???
        \dimen0=\hsize
        \advance\dimen0 by -3em
        \ifdim \wd\@tempboxa >\dimen0
                \hbox to \hsize{
                        \parindent 0em
                        \hfil
                        \parbox{\dimen0}{\def\baselinestretch{0.96}\small
                                {\bf #1.} #2
                                }
                        \hfil}
        \else \hbox to \hsize{\hfil \box\@tempboxa \hfil}
        \fi
        }
\makeatother

\makeatletter \long\def\barenote#1{
    \insert\footins{\footnotesize
    \interlinepenalty\interfootnotelinepenalty
    \splittopskip\footnotesep
    \splitmaxdepth \dp\strutbox \floatingpenalty \@MM
    \hsize\columnwidth \@parboxrestore
    {\rule{\z@}{\footnotesep}\ignorespaces
      #1\strut}}}
\makeatother

\long\def\comment#1{}

\newlength{\widebarargwidth}
\newlength{\widebarargheight}
\newlength{\widebarargdepth}
\DeclareRobustCommand{\widebar}[1]{%
  \settowidth{\widebarargwidth}{\ensuremath{#1}}%
  \settoheight{\widebarargheight}{\ensuremath{#1}}%
  \settodepth{\widebarargdepth}{\ensuremath{#1}}%
  \addtolength{\widebarargwidth}{-0.3\widebarargheight}%
  \addtolength{\widebarargwidth}{-0.3\widebarargdepth}%
  \makebox[0pt][l]{\hspace{0.3\widebarargheight}%
    \hspace{0.3\widebarargdepth}%
    \addtolength{\widebarargheight}{0.3ex}%
    \rule[\widebarargheight]{0.95\widebarargwidth}{0.1ex}}%
  {#1}}

\newcommand{\matsnorm}[2]{|\!|\!| #1 | \! | \!|_{{#2}}}

\newcommand{\graph}{\ensuremath{G}}
\newcommand{\edge}{\ensuremath{E}}
\newcommand{\X}{\ensuremath{x}}
\newcommand{\Qmat}{\ensuremath{P}}
\newcommand{\Qmattil}{\ensuremath{\widetilde{\Qmat}}}
\newcommand{\Lmat}{\ensuremath{L}}
\newcommand{\eigval}{\ensuremath{\lambda}}

\newcommand{\Xbar}{\ensuremath{\widebar{\X}}}
\newcommand{\Xmax}{\ensuremath{\X_{\operatorname{max}}}}
\newcommand{\itnum}{\ensuremath{n}}
\newcommand{\ones}{\ensuremath{\vec{1}}}

\newcommand{\pest}[1]{\ensuremath{\theta^{#1}}}
\newcommand{\pstar}{\ensuremath{\theta^*}}
\newcommand{\Uvar}[1]{\ensuremath{Y^{#1}}}
\newcommand{\Zvar}[1]{\ensuremath{Z^{#1}}}
\newcommand{\Xivar}[1]{\ensuremath{\xi^{#1}}}

\newcommand{\Gfun}[1]{\ensuremath{\mathcal{F}^{#1}}}
\newcommand{\Condcov}[1]{\ensuremath{\Sigma^{#1}}}

\newcommand{\Nset}{\ensuremath{N'}}

\newcommand{\Nfull}{\ensuremath{N}}

\newcommand{\Vmat}{\ensuremath{U}}
\newcommand{\Vtil}{\ensuremath{\widetilde{\Vmat}}}
\newcommand{\Dlmat}{\ensuremath{J}}
\newcommand{\Dlmattil}{\ensuremath{\widetilde{\Dlmat}}}
\newcommand{\diag}{\ensuremath{\operatorname{diag}}}

\newcommand{\Adjac}{\ensuremath{A}}
\newcommand{\Degmat}{\ensuremath{D}}

\newcommand{\zeros}{\ensuremath{\vec{0}}}

\newcommand{\Quant}{\ensuremath{Q}}

\newcommand{\aen}{\ensuremath{\operatorname{AEN}}}
\newcommand{\ann}{\ensuremath{\operatorname{ANN}}}

\newcommand{\bcm}{\ensuremath{\operatorname{BC}}}
\newcommand{\bcors}{\begin{corollary}}
\newcommand{\ecors}{\end{corollary}}
\newcommand{\trace}{\ensuremath{\operatorname{trace}}}
\newcommand{\amse}{\ensuremath{\operatorname{AMSE}}}
\newcommand{\nvar}{\ensuremath{\sigma^2}}
\newcommand{\nquant}{\ensuremath{\sigma^2_{\operatorname{qnt}}}}

\newcommand{\Lapgraph}[1]{\ensuremath{L(#1)}}

\newcommand{\Lstar}{\ensuremath{R(\graph)}}

\newcommand{\pestcts}[1]{\ensuremath{\theta_{#1}}}
\newcommand{\satfun}[1]{\ensuremath{C_M(#1)}}

\newcommand{\gammacts}[1]{\ensuremath{\gamma_{#1}}}
\newcommand{\gamstar}{\ensuremath{\gamma^*}}

\newcommand{\trans}{\ensuremath{t}}
\newcommand{\rec}{\ensuremath{r}}

\newcommand{\stepsize}{\epsilon}

\newcommand{\Ring}{\ensuremath{C}}

\newcommand{\Lattice}{\ensuremath{F}}

\newcommand{\Fully}{\ensuremath{K}}
\newcommand{\degmax}{\ensuremath{d}}

\newcommand{\mydefn}{\ensuremath{: \,=}}

\begin{document}

%%%%%%%%%%%%%%%%

\ifthenelse{\equal{\doctype}{TECH}}
% Stat title page
{
% BEGIN COMMENT
%\comment{
\renewcommand{\baselinestretch}{1.05}
% Hack to make renewcommand take effect immediately (if not in preamble,
% only takes effect after a switch in font size)
\small \normalsize

\begin{center}

{\bf{\LARGE{Network-based consensus averaging with general noisy
channels}}}

\vspace*{.2in}

\begin{tabular}{ccc}
Ram Rajagopal$^\ast$ & \hspace*{.2in} & Martin J. Wainwright$^{\ast,\dagger}$ \\
\texttt{ramr@eecs.berkeley.edu} &  & 
\texttt{wainwrig@stat.berkeley.edu}
\end{tabular}

\begin{tabular}{cc}
Department of Statistics$^\dagger$, and \\
Department of Electrical Engineering and Computer Sciences$^\ast$ \\
University of California, Berkeley \\
Berkeley, CA  94720
\end{tabular}

\vspace*{.1in}

Technical Report \\
Department of Statistics, UC Berkeley \\

\end{center}

\vspace{.2cm}

%\maketitle

\begin{abstract}
This paper focuses on the consensus averaging problem on graphs under
general noisy channels.  We study a particular class of distributed
consensus algorithms based on damped updates, and using the ordinary
differential equation method, we prove that the updates converge
almost surely to exact consensus for finite variance noise.  Our
analysis applies to various types of stochastic disturbances,
including errors in parameters, transmission noise, and quantization
noise.  Under a suitable stability condition, we prove that the error
is asymptotically Gaussian, and we show how the asymptotic covariance
is specified by the graph Laplacian.  For additive parameter noise,
we show how the scaling of the asymptotic MSE is controlled by the
spectral gap of the Laplacian.
\end{abstract}


}
%%%%%%%%%%%%%%%%%%%%%%%%%%%%%%%%
{

\title{Network-based consensus averaging with general noisy channels}

\author{Ram Rajagopal$^1$, Martin Wainwright$^{1,2}$\\
$^1$ Department of Electrical Engineering and Computer Science \\
$^2$ Department of Statistics \\
University of California, Berkeley\\
\texttt{\small $\{$ramr,wainwrig$\}$@eecs.berkeley.edu} \\
}

\maketitle
}
%%%%%%%%%%%%%%%%%%%%%%%%%%%%%%%%%%%%%%%%%%%%%%%%%%%%%%%%%%%%%%%%%%%%%%%%%%%%%%

\section{Introduction}

Consensus problems, in which a group of nodes want to arrive at a
common decision in a distributed manner, have a lengthy history,
dating back to seminal work from over twenty years
ago~\cite{deGroot74, BorVar82,Tsitsiklis84}.  A particular type of
consensus estimation is the distributed averaging problem, in which a
group of nodes want to compute the average (or more generally, a
linear function) of a set of values.  Due to its applications in
sensor and wireless networking, this distributed averaging problem has
been the focus of substantial recent research.  The distributed
averaging problem can be studied either in
continuous-time~\cite{Olfati07}, or in the discrete-time setting
(e.g.,~\cite{Kempe03,Xiao04,Boyd06,AysCoaRab07,DimSarWai08}).  In both
cases, there is now a fairly good understanding of the conditions
under which various distributed averaging algorithms converge, as well
as the rates of convergence for different graph structures.

The bulk of early work on consensus has focused on the case of perfect
communication between nodes.  Given that noiseless communication may
be an unrealistic assumption for sensor networks, a more recent line
of work has addressed the issue of noisy communication links.  With
imperfect observations, many of the standard consensus protocols might
fail to reach an agreement. Xiao et al.~\cite{Xiao07} observed this
phenomenon, and opted to instead redefine the notion of agreement,
obtaining a protocol that allows nodes to obtain a steady-state
agreement, whereby all nodes are able to track but need not obtain
consensus agreement.  Schizas et al.~\cite{Schizas08} study
distributed algorithms for optimization, including the consensus
averaging problem, and establish stability under noisy updates, in
that the iterates are guaranteed to remain within a ball of the
correct consensus, but do not necessarily achieve exact consensus.
Kashyap et al.~\cite{Kashyap07} study consensus updates with the
additional constraint that the value stored at each node must be
integral, and establish convergence to quantized consensus.  Fagnani
and Zampieri~\cite{FagZam07} study the case of packet-dropping
channels, and propose various updates that are guaranteed to achieve
consensus.  Yildiz and Scaglione~\cite{YilSca07} suggest coding
strategies to deal with quantization noise, but do not establish
convergence.  In related work, Aysal et al~\cite{AysCoaRab07} used
probabilistic forms of quantization to develop algorithms that achieve
consensus in expectation, but not in an almost sure sense.

In the current paper, we address the discrete-time average consensus
problem for general stochastic channels.  Our main contribution is to
propose and analyze simple distributed protocols that are guaranteed
to achieve exact consensus in an almost sure (sample-path) sense.
These exactness guarantees are obtained using protocols with
decreasing step sizes, which smooths out the noise factors.  The
framework described here is based on the classic ordinary differential
equation method~\cite{Ljung77}, and allows for the analysis of several
different and important scenarios, namely:
\begin{itemize}
\item Noisy storage: stored values at each node are corrupted by
noise, with known covariance structure.
\item Noisy transmission:  messages across each edge are
corrupted by noise,  with known covariance structure.
\item Bit constrained channels: dithered quantization is applied to
messages prior to transmission.
\end{itemize}
To the best of our knowledge, this is the first paper to analyze
protocols that can achieve arbitrarily small mean-squared error (MSE)
for distributed averaging with noise.  By using stochastic
approximation theory~\cite{Benveniste90,Kushner97}, we establish
almost sure convergence of the updates, as well as asymptotic
normality of the error under appropriate stability conditions. The
resulting expressions for the asymptotic variance reveal how different
graph structures---ranging from ring graphs at one extreme, to
expander graphs at the other---lead to different variance scaling
behaviors, as determined by the eigenspectrum of the graph
Laplacian~\cite{Chung}.

The remainder of this paper is organized as follows.  We begin in
Section~\ref{SecBackground} by describing the distributed averaging
problem in detail, and defining the class of stochastic algorithms
studied in this paper.  In Section~\ref{SecMain}, we state our main
results on the almost-sure convergence and asymptotic normality of our
protocols, and illustrate some of their consequences for particular
classes of graphs.  In particular, we illustrate the sharpness of our
theoretical predictions by comparing them to simulation results, on
various classes of graphs. Section~\ref{SecProof} is devoted to the
proofs of our main results, and we conclude the paper with discussion
in Section~\ref{SecDiscuss}.  (This work was presented in part at the
Allerton Conference on Control, Computing and Communication in
September 2007.)  \\

\vspace*{.1in}

\noindent {\bf{Comment on notation:}} Throughout this paper, we use
the following standard asymptotic notation: for a functions $f$ and
$g$, the notation $f(n) = \mathcal{O}(g(n))$ means that $f(n) \leq C
g(n)$ for some constant $C < \infty$; the notation $f(n) =
\Omega(g(n))$ means that $f(n) \geq C' g(n)$ for some constant $C' >
0$, and $f(n) = \Theta(g(n))$ means that $f(n) = \mathcal{O}(g(n))$
and $f(n) = \Omega(g(n))$.

\section{Problem set-up}
\label{SecBackground}

In this section, we describe the distributed averaging problem, and
specify the class of stochastic algorithms studied in this paper.

\subsection{Consensus matrices and stochastic updates}

Consider a set of $\nodenum = |\nodeSet|$ nodes, each representing a
particular sensing and processing device.  We model this system as an
undirected graph $\graph = (\nodeSet, \edge)$, with processors
associated with nodes of the graph, and the edge set $\edge \subset
\nodeSet \times \nodeSet$ representing pairs of processors that can
communicate directly.  For each node $\vertex$, we let
$\Nfull(\vertex) \defn \{ u \in \nodeSet \, \mid \, (\vertex, u) \in
\edge \}$ be its neighborhood set.

\begin{figure}[h]
\begin{center}
\psfrag{#th#}{$\pest{}(t)$}
\psfrag{#t3#}{$t$}
\psfrag{#t2#}{$u$}
\psfrag{#t1#}{$s$}
\psfrag{#th1#}{$\pest{}(s)$}
\psfrag{#th2#}{$\pest{}(u)$}
\psfrag{#r#}{$r$}
\psfrag{#Ffun#}{$\Gfun{}(\pest{t}, \Xivar{}(t,r))$}
\widgraph{.34\textwidth}{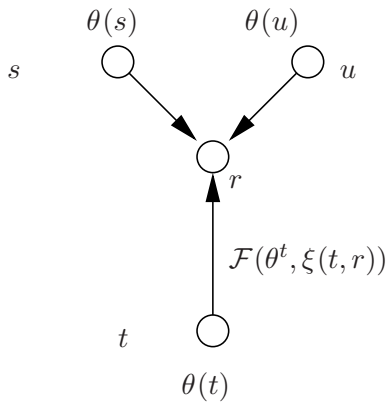}
\caption{Illustration of the distributed protocol.  Each node $t \in V$
maintains an estimate $\pest{}(t)$.  At each round, for a fixed reference
node $r \in V$, each neighbor $t \in N(r)$ sends the message
$\Gfun{}(\pest{t}, \Xivar{}(t,r))$ along the edge $t \rightarrow r$.}
\label{FigNetwork}
\end{center}
\end{figure}
Suppose that each vertex $\vertex$ makes a real-valued measurement
$\X(\vertex)$, and consider the goal of computing the average $\Xbar =
\frac{1}{\nodenum} \sum_{\vertex \in \nodeSet} \X(v)$.  We assume that
$|\X(\vertex)| \leq \Xmax$ for all $\vertex \in \nodeSet$, as dictated
by physical constraints of sensing.  For iterations $\itnum = 0,1,2,
\ldots$, let $\pest{\itnum} = \{\pest{\itnum}(\vertex), \; \vertex \in
\nodeSet \}$ represent an $\nodenum$-dimensional vector of estimates.
Solving the distributed averaging problem amounts to having
$\pest{\itnum}$ converge to $\pstar \defn \Xbar \, \ones$, where
$\ones \in \real^\nodenum$ is the vector of all ones.  Various
algorithms for distributed averaging~\cite{Olfati07,Boyd06} are based
on symmetric consensus matrices $\Lmat \in \real^{\numnode \times
\numnode}$ with the properties:
\begin{subequations}
\label{EqnDefnConsensus}
\begin{eqnarray}
\Lmat(\vertex, \vertex') & \neq & 0 \quad \mbox{only if $(\vertex,
\vertex') \in \edge$} \\
\Lmat \ones & = & \zeros, \qquad \mbox{and} \\
\Lmat & \succeq & 0.
\end{eqnarray}
\end{subequations}
The simplest example of such a matrix is the \emph{graph Laplacian},
defined as follows.  Let $\Adjac \in \real^{\numnode \times \numnode}$
be the adjacency matrix of the graph $\graph$, i.e. the symmetric
matrix with entries
\begin{eqnarray}
\Adjac_{ij} & = & \begin{cases} 1 & \mbox{if $(i,j) \in \edge$} \\
0 & \mbox{otherwise,}
          \end{cases}
\end{eqnarray}
and let $\Degmat = \diag \{d_1, d_2, \ldots, d_\numnode \}$ where $d_i
\defn |\Nfull(i)|$ is the degree of node $i$.  Assuming that the graph
is connected (so that $d_i \geq 1$ for all $i$), the graph Laplacian
is given by
\begin{eqnarray}
\label{EqnGraphLap}
\Lapgraph{\graph} & = & I - \Degmat^{-1/2} \Adjac \Degmat^{-1/2}.
\end{eqnarray}
Our analysis applies to the (rescaled) graph Laplacian, as well as to
various weighted forms of graph Laplacian matrices~\cite{Chung}.

Given a fixed choice of consensus matrix $\Lmat$, we consider the
following family of updates, generating the sequence $\{\pest{\itnum},
\; \itnum = 0, 1,2 \ldots \}$ of $\nodenum$-dimensional vectors.  The
updates are designed to respect the neighborhood structure of the
graph $\graph$, in the sense that at each iteration, the estimate
$\pest{\itnum+1}(\rec)$ at a \emph{receiving node} $r \in \nodeSet$ is
a function of only\footnote{In fact, our analysis is easily
generalized to the case where $\pest{\itnum+1}(\rec)$ depends only on
vertices $\trans \in \Nset(\rec)$, where $\Nset(\rec)$ is a (possibly
random) subset of the full neighborhood set $\Nfull(\vertex)$.
However, to bring our results into sharp focus, we restrict attention
to the case $\Nset(\rec) = \Nfull(\rec)$.}  the estimates
\mbox{$\{\pest{\itnum}(\trans), \; \trans \in \Nfull(\rec) \}$}
associated with \emph{transmitting nodes} $\trans$ in the neighborhood
of node $\rec$.  In order to model noise and uncertainty in the
storage and communication process, we introduce random variables
$\Xivar{}(\trans, \rec)$ associated with the transmission link from
$\trans$ to $\rec$; we allow for the possibility that
$\Xivar{}(\trans, \rec) \neq \Xivar{}(\rec, \trans)$, since the noise
structure might be asymmetric.

With this set-up, we consider algorithms that generate a stochastic
sequence $\{\pest{\itnum}, \itnum = 0, 1, 2, \ldots \}$ in the
following manner:
\vspace*{.1in}

\framebox[.99\textwidth]{\parbox{.95\textwidth}{
\begin{enumerate}
\item[1.] At time step $\itnum = 0$, initialize $\pest{0}(\vertex) =
\X(\vertex)$ for all $\vertex \in \nodeSet$.
\item[2.]  For time steps $\itnum = 0, 1,2, \ldots$, each node $\trans
\in \nodeSet$ computes the random variables
\begin{equation}
\label{EqnDefnUvar}
\Uvar{\itnum+1}(\rec, \trans) =
\begin{cases}
\pest{\itnum}(\trans), & \mbox{if $\trans = \rec$} \\
  \Gfun{}(\pest{\itnum}(\trans),\Xivar{\itnum+1}(\trans,\rec)) &
  \mbox{if $(\trans, \rec) \in \edge$} \\
0 & \mbox{otherwise},
                        \end{cases}
\end{equation}
where $\Gfun{}$ is the \emph{communication-noise function} defining
the model.
\item[3.] Generate estimate $\pest{\itnum+1} \in \real^\nodenum$ as
\begin{equation}
\label{EqnCompactUpdate}
\pest{\itnum+1} \, = \, \pest{\itnum} +   \stepsize_\itnum
\left[-\left(\Lmat \odot \Uvar{\itnum+1} \right) \ones \right],
\end{equation}
where $\odot$ denotes the Hadamard (elementwise) product between
matrices, and $\stepsize_\itnum > 0$ is a decaying step size
parameter.
\end{enumerate}
}}

\vspace*{.2in} 

See Figure~\ref{FigNetwork} for an illustration of the
message-passing update of this protocol.  In this paper, we focus on
step size parameters $\stepsize_\itnum$ that scale as
$\stepsize_\itnum = \Theta(1/\itnum)$.  On an elementwise basis, the
update~\eqref{EqnCompactUpdate} takes the form
\begin{eqnarray*}
\pest{\itnum+1}(\rec) & = & \pest{\itnum}(\rec) - \stepsize_\itnum
\left[ \Lmat(\rec, \rec) \pest{\itnum}(\rec) + \sum_{\trans \in
\Nfull(\rec)} \Lmat(\rec, \trans) \;
\Gfun{}(\pest{\itnum}(\trans),\Xivar{\itnum+1}(\trans,\rec)) \right].
\end{eqnarray*}

\subsection{Communication and noise models}
\label{SecCommModel}

It remains to specify the form of the the function $\Gfun{}$ that
controls the communication and noise model in the local computation
step in equation~\eqref{EqnDefnUvar}. \\

\vspace*{.05in}

\noindent \myparagraph{Noiseless real number model:} The simplest
model, as considered by the bulk of past work on distributed
averaging, assumes noiseless communication of real numbers.  This
model is a special case of the update~\eqref{EqnDefnUvar} with
$\Xivar{\itnum}(\trans, \rec) = 0$, and
\begin{equation}
\Gfun{}(\pest{\itnum}(\trans),\Xivar{\itnum+1}(\trans,\rec)) =
\pest{\itnum}(\trans).
\end{equation}

\vspace*{.02in}

\noindent \myparagraph{Additive edge-based noise model ($\aen$):} In
this model, the term $\Xivar{\itnum}(\trans, \rec)$ is zero-mean
additive random noise variable that is associated with the
transmission $\trans \rightarrow \rec$, and the communication function
takes the form
\begin{equation}
\label{EqnAEN}
\Gfun{}(\pest{\itnum}(\trans),\Xivar{\itnum+1}(\trans, \rec)) =
\pest{\itnum}(\trans) + \Xivar{\itnum+1}(\trans, \rec).
\end{equation}
We assume that the random variables $\Xivar{\itnum+1}(\trans, \rec)$
and $\Xivar{\itnum+1}(\trans', \rec)$ are independent for
distinct edges $(\trans', \rec)$ and $(\trans, \rec)$, and identically
distributed with zero-mean and variance
$\sigma^2 = \var(\Xivar{\itnum+1}(\trans, \rec))$.\\

\vspace*{.02in}

\noindent \myparagraph{Additive node-based noise model ($\ann$):} In
this model, the function $\Gfun{}$ takes the same form~\eqref{EqnAEN}
as the edge-based noise model.  However, the key distinction is that
for each $\vertex' \in \nodeSet$, we assume that
\begin{eqnarray}
\label{EqnANN}
\Xivar{\itnum+1}(\trans, \rec) & = & \Xivar{\itnum+1}(\trans)
\qquad \mbox{for all $\rec \in \Nfull(\trans)$,}
\end{eqnarray}
where $\Xivar{\itnum+1}(\trans)$ is a single noise variable associated
with node $\trans$, with zero mean and variance $\sigma^2 =
\var(\Xivar{\itnum}(\trans))$. Thus, the random variables
$\Xivar{\itnum+1}(\trans, \rec)$ and $\Xivar{\itnum+1}(\trans, \rec')$
are all \emph{identical} for all edges out-going from the transmitting
node $\trans$. \\

\vspace*{.02in}

\noindent \myparagraph{Bit-constrained communication ($\bcm$):}
Suppose that the channel from node $\vertex'$ to $\vertex$ is
bit-constrained, so that one can transmit at most $B$ bits, which is
then subjected to random dithering.  Under these assumptions, the
communication function $\Gfun{}$ takes the form
\begin{eqnarray}
\label{EqnBCM}
\Gfun{}(\pest{}(\vertex'),\Xivar{}(\vertex',\vertex)) & = & \Quant_B
\left(\pest{}(\vertex')+\Xivar{}(\vertex', \vertex) \right),
\end{eqnarray}
where $\Quant_B(\cdot)$ represents the $B$-bit quantization function
with maximum value $M$ and $\Xivar{}(\vertex', \vertex)$ is random
dithering.  We assume that the random dithering is applied prior to
transmission across the channel out-going from vertex $\vertex'$, so
that  $\Xivar{}(\vertex', \vertex) = \Xivar{}(\vertex')$ is the
same random variable across all neighbors $\vertex \in \Nfull(\vertex')$.

\section{Main result and consequences}
\label{SecMain}

In this section, we first state our main result, concerning the
stochastic behavior of sequence $\{\pest{\itnum} \}$ generated by
the updates~\eqref{EqnCompactUpdate}.  We then illustrate its
consequences for the specific communication and noise models
described in Section~\ref{SecCommModel}, and conclude with a
discussion of behavior for specific graph structures.

\subsection{Statement of main result}

Consider the factor $\Lmat \odot \Uvar{}$ that drives the
updates~\eqref{EqnCompactUpdate}.  An important element of our
analysis is the conditional covariance of this update factor, denoted
by $\Condcov{} = \Condcov{}_{\pest{}}$ and given by
\begin{eqnarray}
\label{EqnDefnCondcov}
\Condcov{}_{\pest{}} & \defn &
\expec\left[\left(\Lmat\,\odot\,\Uvar{}(\pest{},
\Zvar{})\right)\,\ones \; \ones^T \, \left(\Lmat\,\odot \,
\Uvar{}(\pest{}, \Zvar{}) \right)^T \; \mid \; \pest{} \right] - \Lmat
\, \pest{} (\Lmat \pest{})^T.
\end{eqnarray}
A little calculation shows that the $(i,j)^{th}$ element of this
matrix is given by
\begin{eqnarray}
\label{EqnDefnInner}
\Condcov{}_{\pest{}}(i,j) & = & \sum_{k,\ell=1}^\numnode \Lmat(i,k)
\Lmat(j,\ell) \; \, \expec\left[\Uvar{}(i,k) \Uvar{}(j,\ell) -
\pest{}(k) \pest{}(\ell) \; \mid \; \pest{} \right].
\end{eqnarray}

Moreover, the eigenstructure of the consensus matrix $\Lmat$ plays an
important role in our analysis.  Since it is symmetric and positive
semidefinite, we can write
\begin{eqnarray}
\label{EqnDefnVmat}
\Lmat & = &  \Vmat \Dlmat  \Vmat^T,
\end{eqnarray}
where $\Vmat$ be an $\numnode \times \numnode$ orthogonal matrix with
columns defined by unit-norm eigenvectors of $\Lmat$, and $\Dlmat
\mydefn \diag \{\eigval_1(\Lmat), \ldots, \eigval_\numnode(\Lmat) \}$ is
a diagonal matrix of eigenvalues, with
\begin{eqnarray}
\label{EqnDefnEigs}
0 \; = \; \eigval_1(\Lmat) < \eigval_2(\Lmat) \leq \ldots <
\eigval_\numnode(\Lmat).
\end{eqnarray}
It is convenient to let $\Vtil$ denote the $\numnode \times (\numnode
-1)$ matrix with columns defined by eigenvectors associated with
positive eigenvalues of $\Lmat$--- that is, excluding column $\Vmat_1
= \ones/\|\ones\|_2$, associated with the zero-eigenvalue
$\eigval_1(\Lmat) = 0$.  With this notation, we have
\begin{eqnarray}
\label{EqnDefnDlmattil}
\Dlmattil \; = \; \diag \{ \eigval_2(\Lmat), \ldots,
\eigval_\numnode(\Lmat) \} & = &  \Vtil^T  \Lmat \Vtil.
\end{eqnarray}

\begin{theorem}
\label{ThmConsist}
Consider the random sequence $\{\pest{\itnum}\}$ generated by the
update~\eqref{EqnCompactUpdate} for some communication function
$\Gfun{}$, consensus matrix $\Lmat$, and step size parameter
$\stepsize_\itnum = \Theta(1/\itnum)$.
\begin{itemize}
\item[(a)] In all cases, the sequence $\{\pest{\itnum}\}$ is a
strongly consistent estimator of $\pstar = \Xbar \ones$, meaning that
$\pest{\itnum} \rightarrow \pstar$ almost surely (a.s.).

\item[(b)] Furthermore, if the second smallest eigenvalue of the
consensus matrix $\Lmat$ satisfies $\eigval_2(\Lmat) > 1/2$ then
\begin{align}
\sqrt{n}(\theta_{n}-\theta^*) \convdist N \left(0, \Vmat^T
\begin{bmatrix} 0 & 0 \\0 & \Qmattil \end{bmatrix} \Vmat \right),
\end{align}
where the $(\numnode-1) \times (\numnode-1)$ matrix $\Qmattil$ is the
solution of the continuous time Lyapunov equation
\begin{eqnarray}
\label{EqnLyap}
\left(\Dlmattil - \frac{I}{2}\right) \Qmattil
+\Qmattil\left(\Dlmattil-\frac{I}{2}\right)^T & = &
\widetilde{\Condcov{}}_{\pstar}
\end{eqnarray}
where $\Dlmattil$ is the diagonal matrix~\eqref{EqnDefnDlmattil}, and
$\widetilde{\Condcov{}}_{\pstar} = \Vtil^T \Condcov{}_{\pstar} \Vtil$
is the transformed version of the conditional
covariance~\eqref{EqnDefnCondcov}.
\end{itemize}
\end{theorem}
\bigskip
Theorem~\ref{ThmConsist}(a) asserts that the sequence
$\{\pest{\itnum}\}$ is a strongly consistent estimator of the average.
As opposed to weak consistency, this result guarantees that for almost
any realization of the algorithm, the associated sample path converges
to the exact consensus solution.  Theorem~\ref{ThmConsist}(b)
establishes that for appropriate choices of consensus matrices, the
rate of MSE convergence is of order $1/\itnum$, since the
$\sqrt{\itnum}$-rescaled error converges to a non-degenerate Gaussian
limit.  Such a rate is to be expected in the presence of sufficient
noise, since the number of observations received by any given node
(and hence the inverse variance of estimate) scales as $\itnum$.  The
solution of the Lyapunov equation~\eqref{EqnLyap} specifies the
precise form of this asymptotic covariance, which (as we will see)
depends on the graph structure.

\subsection{Some consequences}

Theorem~\ref{ThmConsist} can be specialized to particular noise and
communication models.  Here we derive some of its consequences for the
$\aen$, $\ann$ and $\bcm$ models.  For any model for which
Theorem~\ref{ThmConsist}(b) holds, we define the average mean-squared
error as
\begin{eqnarray}
\amse(\Lmat; \pstar) & \defn & \frac{1}{\numnode}
\trace(\Qmattil(\pstar)),
\end{eqnarray}
corresponding to asymptotic error variance, averaged over nodes of the
graph.

\bcors[Asymptotic MSE for specific models]
\label{CorAMSE}
Given a consensus matrix $\Lmat$ with second-smallest eigenvalue
\mbox{$\eigval_2(\Lmat) > \frac{1}{2}$}, the sequence
$\{\pest{\itnum}\}$ is a strongly consistent estimator of the average
$\pstar$, with asymptotic MSE characterized as follows:
\begin{enumerate}
\item[(a)] For the additive edge-based noise ($\aen$)
model~\eqref{EqnAEN}:
\begin{eqnarray}
\label{EqnAENmse}
\amse(\Lmat; \pstar) & \leq & \frac{\nvar}{\numnode} \;
\sum_{i=2}^\numnode \left[\frac{\max \limits_{j=1, \ldots, \numnode}
\sum_{k \neq j} \Lmat^2(j,k)}{2 \eigval_i(\Lmat) - 1} \right].
\end{eqnarray}
\item[(b)] For the additive node-based noise ($\ann$)
model~\eqref{EqnANN} and the bit-constrained ($\bcm$)
model~\eqref{EqnBCM}:
\begin{eqnarray}
\label{EqnANNBCmse}
\amse(\Lmat; \pstar) & = & \frac{\sigma^2}{\numnode} \; \;
\sum_{i=2}^\numnode \left[\frac{[\eigval_i(\Lmat)]^2}{2
\eigval_i(\Lmat) - 1} \right],
\end{eqnarray}
where the variance term $\sigma^2$ is given by the quantization noise
$\expec\left[\Quant_B(\pest{}+\Xivar{})^2- \theta^2 \, \mid \,
\pest{}\right]$ for the $\bcm$ model, and the noise variance
$\var(\Xivar{}(\vertex'))$ for the $\ann$ model.
\end{enumerate}
\ecors
\begin{proof}
The essential ingredient controlling the asymptotic MSE is the
conditional covariance matrix $\Condcov{}_{\pstar}$, which specifies
$\Qmattil$ via the Lyapunov equation~\eqref{EqnLyap}.  For analyzing
model $\aen$, it is useful to establish first the
following auxiliary result.  For each $i=1, \ldots, \numnode-1$, we
have
\begin{eqnarray}
\label{EqnDiagInequal}
\Qmattil_{ii} & \leq & \frac{\matsnorm{\Condcov{}_{\pstar}}{2}{}}{ 2
\, \eigval_{i+1}(\Lmat) - 1},
\end{eqnarray}
where $\matsnorm{\Condcov{}_{\pstar}}{2}{} =
\matsnorm{\Condcov{}}{2}{}$ is the spectral norm (maximum eigenvalue
for a positive semdefinite symmetric matrix).  To see this fact, note
that
\begin{equation*}
\Vtil^T \Condcov{} \Vtil \; \preceq \; \Vtil^T
\left[\matsnorm{\Condcov{}}{2}{} I \right] \Vtil \; = \;
\matsnorm{\Condcov{}}{2}{} I.
\end{equation*}
Since $\Qmattil$ satisfies the Lyapunov equation, we have
\begin{eqnarray*}
\left(\Dlmattil - \frac{I}{2}\right) \Qmattil
+\Qmattil\left(\Dlmattil-\frac{I}{2}\right)^T & \preceq &
\matsnorm{\Condcov{}}{2}{} I.
\end{eqnarray*}
Note that the diagonal entries of the matrix $\left(\Dlmattil -
\frac{I}{2}\right) \Qmattil
+\Qmattil\left(\Dlmattil-\frac{I}{2}\right)^T$ are of the form $(2
\eigval_{i+1} - 1) \; \Qmattil_{ii}$.  The difference between the RHS
and LHS matrices constitute a positive semidefinite matrix, which must
have a non-negative diagonal, implying the claimed
inequality~\eqref{EqnDiagInequal}.

In order to use the bound~\eqref{EqnDiagInequal}, it remains to
compute or upper bound the spectral norm $\matsnorm{\Condcov{}}{2}{}$,
which is most easily done using the elementwise
representation~\eqref{EqnDefnInner}. \\

\noindent (a) For the $\aen$ model~\eqref{EqnAEN}, we have
\begin{eqnarray}
\expec\left[\Uvar{}(i,k)\Uvar{}(j,\ell) - \pest{}_k \pest{}_\ell \,
  \mid  \, \pest{}\right] & = & \expec\left[\Xivar{}(i,k)
  \Xivar{}(j,\ell)\right].
\end{eqnarray}
Since we have assumed that the random variables $\Xivar{}(i,k)$ on
each edge $(i,k)$ are i.i.d., with zero-mean and variance $\nvar$, we
have
\begin{eqnarray*}
\expec\left[\Uvar{}(i,k)\Uvar{}(j,\ell) - \pest{}(k) \pest{}(\ell) \,
  \mid \, \pest{}\right] & = & \begin{cases} \nvar & \mbox{if $(i,k) =
  (j, \ell)$ and $i \neq j$} \\
0 & \mbox{otherwise.}
                   \end{cases}
\end{eqnarray*}
Consequently, from the elementwise expression~\eqref{EqnDefnInner}, we conclude
that $\Condcov{}$ is diagonal, with entries
\begin{eqnarray*}
\Condcov{}(j,j) & = & \nvar \; \sum_{k \neq j} \Lmat^2(k,j),
\end{eqnarray*}
so that $\matsnorm{\Condcov{}}{2}{} = \nvar \, \max_{j =1, \ldots,
\numnode} \sum_{k \neq j} \Lmat^2_{jk}$, which establishes the
claim~\eqref{EqnAENmse}.  \\

\noindent (b) For the $\bcm$ model~\eqref{EqnBCM}, we have
\begin{eqnarray}
\expec\left[\Uvar{}(i,k)\Uvar{}(j,\ell) - \pest{}(k) \pest{}(\ell) \,
  \mid \, \pest{}\right] & = & \begin{cases} \nquant & \mbox{if $i =j$
  and $k = \ell$} \\
0  & \mbox{otherwise},
                   \end{cases}
\end{eqnarray}
where $\nquant \defn \expec\left[\Quant_B(\pest{}+\Xivar{})^2-
\theta^2 \, \mid \, \pest{}\right]$ is the quantization noise.
Therefore, we have $\Condcov{}(\pstar) = \nquant \Lmat^2$, and using
the fact that $\Vtil$ consists of eigenvectors of $\Lmat$ (and hence
also $\Lmat^2$, the Lyapunov equation~\eqref{EqnLyap} takes the form
\begin{eqnarray*}
\left(\Dlmattil - \frac{I}{2}\right) \Qmattil
+\Qmattil\left(\Dlmattil-\frac{I}{2}\right)^T & = & \nquant \,
(\Dlmattil)^2,
\end{eqnarray*}
which has the explicit diagonal solution $\Qmattil$ with entries
$\Qmattil_{ii} = \frac{\nquant \, \eigval_{i+1}^2(\Lmat)}{2
\eigval_{i+1}(\Lmat) - 1}$.  Computing the asymptotic MSE
$\frac{1}{\numnode} \sum_{i=1}^{\numnode-1} \Qmattil_{ii}$ yields the
claim~\eqref{EqnANNBCmse}.  The proof of the same claim for the $\ann$
model is analogous.
\end{proof}

\subsection{Scaling behavior for specific graph classes}

We can obtain further insight by considering Corollary~\ref{CorAMSE}
for specific graphs, and particular choices of consensus matrices
$\Lmat$.  For a fixed graph $\graph$, consider the graph Laplacian
$\Lapgraph{\graph}$ defined in equation~\eqref{EqnGraphLap}.  It is
easy to see that $\Lapgraph{\graph}$ is always positive semi-definite,
with minimal eigenvalue $\eigval_1(\Lapgraph{\graph}) = 0$,
corresponding to the constant vector.  For a connected graph, the
second smallest eigenvector $\Lapgraph{\graph}$ is strictly
positive~\cite{Chung}.  Therefore, given an undirected graph $\graph$
that is connected, the most straightforward manner in which to obtain
a consensus matrix $\Lmat$ satisfying the conditions of
Corollary~\ref{CorAMSE} is to rescale the graph Laplacian
$\Lapgraph{\graph}$, as defined in equation~\eqref{EqnGraphLap}, by
its second smallest eigenvalue $\eigval_2(\Lapgraph{\graph})$, thereby
forming the rescaled consensus matrix
\begin{eqnarray}
\label{EqnDefnLstar} 
\Lstar &  \defn & \frac{1}{\eigval_2(\Lapgraph{\graph})} \;
\Lapgraph{\graph}.
\end{eqnarray}
with $\eigval_2(\Lstar) = 1 > \frac{1}{2}$. 

 With this choice of consensus matrix, let us consider the
implications of Corollary~\ref{CorAMSE}(b), in application to the
additive node-based noise ($\ann$) model, for various graphs.  We
begin with a simple lemma, proved in Appendix~\ref{AppLemOrder},
showing that, up to constants, the scaling behavior of the asymptotic
MSE is controlled by the second smallest eigenvalue
$\eigval_2(\Lapgraph{\graph})$.
\begin{lemma}
\label{LemOrder}
For any connected graph $\graph$, using the rescaled Laplacian
consensus matrix~\eqref{EqnDefnLstar}, the asymptotic MSE for the
$\ann$ model~\eqref{EqnANN} satisfies the bounds
\begin{equation}
\frac{\sigma^2}{2 \eigval_2(\Lapgraph{\graph})} \; \leq \;
\amse(\Lstar; \pstar) \; \leq \;
\frac{\sigma^2}{\eigval_2(\Lapgraph{\graph})},
\end{equation}
where $\eigval_2(\Lapgraph{\graph})$ is the second smallest
eigenvalue of the graph.
\end{lemma}
Combined with known results from spectral graph theory~\cite{Chung},
Lemma~\ref{LemOrder} allows us to make specific predictions about the
number of iterations required, for a given graph topology of a given
size $\numnode$, to reduce the asymptotic MSE to any $\delta > 0$:
in particular, the required number of iterations scales as
\begin{eqnarray}
\label{EqnKeyScaling}
\itnum & = & \Theta
\left(\frac{\sigma^2}{\eigval_2(\Lapgraph{\graph})} \;
\frac{1}{\delta} \right).
\end{eqnarray}
Note that this scaling is similar but different from the scaling of
noiseless updates~\cite{Boyd06,DimSarWai08}, where the MSE is (with
high probability) upper bounded by $\delta$ for $\itnum =
\Theta(\frac{\log(1/\delta)}{-\log(1-\eigval_2(\Lapgraph{\graph}))})$,
which scales as
\begin{eqnarray}
\itnum & = & \Theta \left(\frac{\log
(1/\delta)}{\eigval_2(\Lapgraph{\graph})} \right),
\end{eqnarray}
for a decaying spectral gap $\eigval_2(\Lapgraph{\graph}) \rightarrow
0$.

\subsection{Illustrative simulations}
We illustrate the predicted scaling~\eqref{EqnKeyScaling} by some
simulations on different classes of graphs.  For all experiments
reported here, we set the step size parameter $\stepsize_\itnum =
\frac{1}{\itnum + 100}$.  The additive offset serves to ensure
stability of the updates in very early rounds, due to the possibly
large gain specified by the rescaled
Laplacian~\eqref{EqnDefnLstar}. We performed experiments for a range
of graph sizes, for the additive node noise ($\ann$)
model~\eqref{EqnANN}, with noise variance $\sigma^2 = 0.1$ in all
cases. For each graph size $\nodenum$, we measured the number of
iterations $\itnum$ required to reach a fixed level $\delta$ of
mean-squared error.  
\subsubsection{Cycle graph} 

Consider the ring graph $\Ring_\numnode$ on $\numnode$ vertices, as
illustrated in Figure~\ref{FigRing}(a).  Panel (b) provides a log-log
plot of the MSE versus the iteration number $\itnum$; each trace
corresponds to a particular sample path.  Notice how the MSE over each
sample coverges to zero.  Moreover, since Theorem~\ref{ThmConsist}
predicts that the MSE should drop off as $1/\itnum$, the linear rate
shown in this log-log plot is consistent.  Figure~\ref{FigRing}(c)
plots the number of iterations (vertical axis) required to achieve a
given constant MSE versus the size of the ring graph (horizontal
axis).  For the ring graph, it can be shown (see Chung~\cite{Chung})
that the second smallest eigenvalue scales as
$\eigval_2(\Lapgraph{\Ring_\numnode}) = \Theta(1/\numnode^2)$, which
implies that the number of iterations to achieve a fixed MSE for a
ring graph with $\numnode$ vertices should scale as $\itnum =
\Theta(\numnode^2)$.  Consistent with this prediction, the plot in
Figure~\ref{FigRing}(c) shows a quadratic scaling; in particular, note
the excellent agreement between the theoretical prediction and the
data.

\newcommand{\figwone}{0.2\textwidth}
\newcommand{\figwtwo}{0.38\textwidth}
\begin{figure}[h]
\begin{center}
\begin{tabular}{ccc}
\raisebox{.2in}{\widgraph{\figwone}{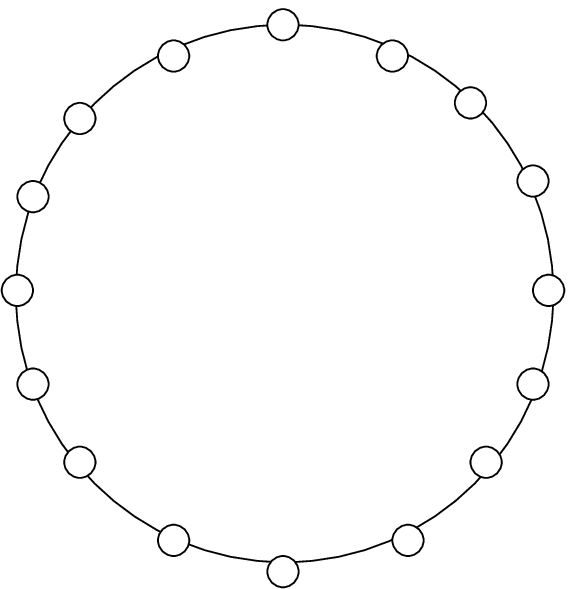}} & 
\widgraph{\figwtwo}{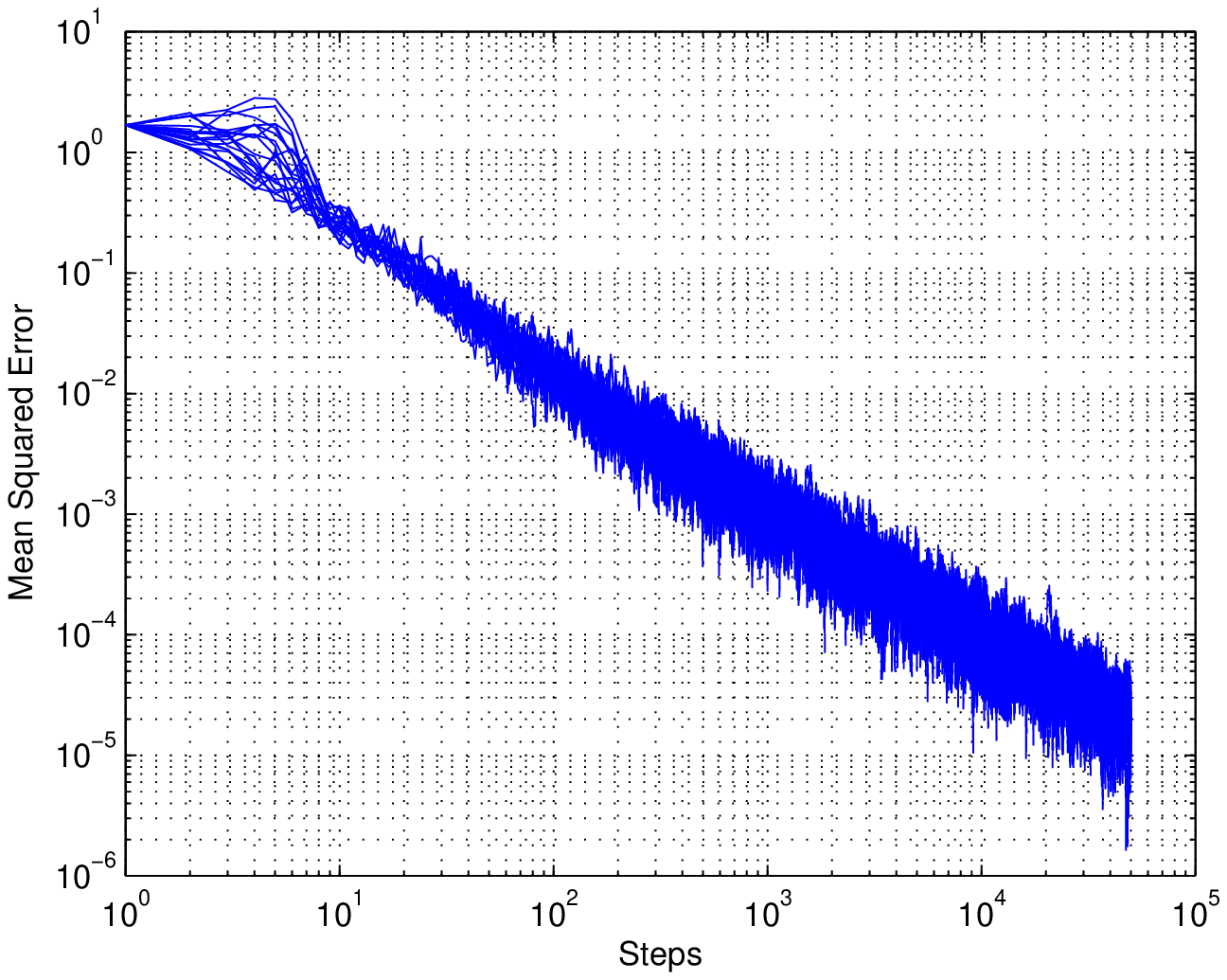} &
\widgraph{\figwtwo}{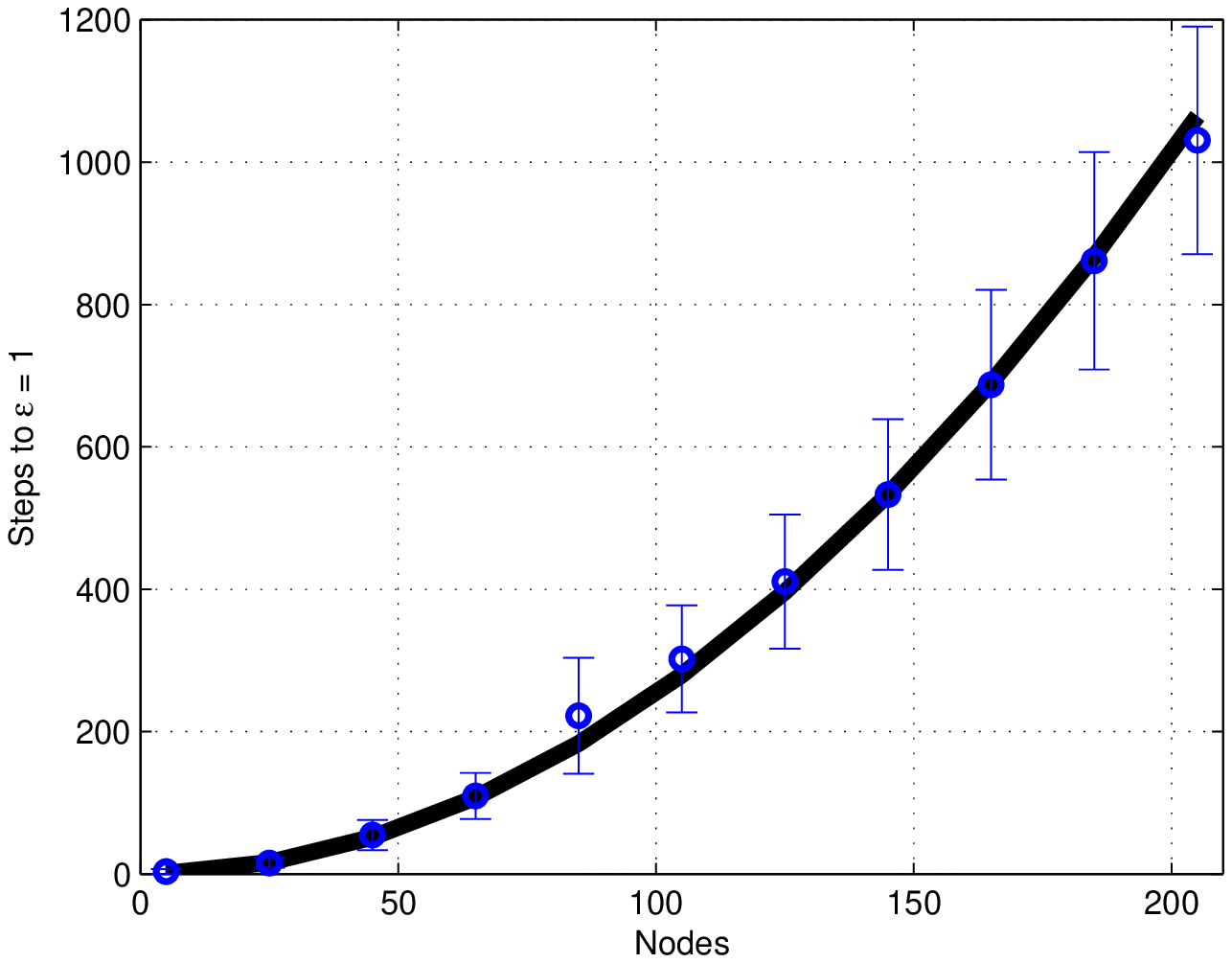} \\
(a) & (b) & (c)
\end{tabular}
\caption{Comparison of empirical simulations to theoretical
predictions for the ring graph in panel (a).  (b) Sample path plots of
log MSE versus log iteration number: as predicted by the theory, the
log MSE scales linearly with log iterations.  (c) Plot of number of
iterations (vertical axis) required to reach a fixed level of MSE
versus the graph size (horizontal axis).  For the ring graph, this
quantity scales quadratically in the graph size, consistent with
Corollary~\ref{CorAMSE}.}
\label{FigRing}
\end{center}
\end{figure}

\subsubsection{Lattice model}

Figure~\ref{FigGrid}(a) shows the two-dimensional four
nearest-neighbor lattice graph with $\numnode$ vertices, denoted
$\Lattice_\numnode$.  Again, panel (b) corresponds to a log-log plot
of the MSE versus the iteration number $\itnum$, with each trace
corresponding to a particular sample path, again showing a linear rate
of convergence to zero.  Panel (c) shows the number of iterations
required to achieve a constant MSE as a function of the graph size.
For the lattice, it is known~\cite{Chung} that
$\eigval_2(\Lmat(\Lattice_\numnode)) = \Theta(1/\numnode)$, which
implies that the critical number of iterations should scale as $\itnum
= \Theta(\numnode)$.  Note that panel (c) shows linear scaling, again
consistent with the theory.

\begin{figure}[h]
\begin{center}
\begin{tabular}{ccc}
\raisebox{.2in}{\widgraph{\figwone}{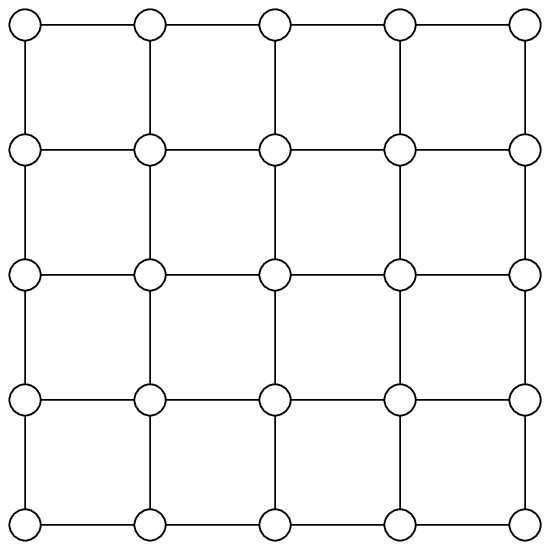}} & 
\widgraph{\figwtwo}{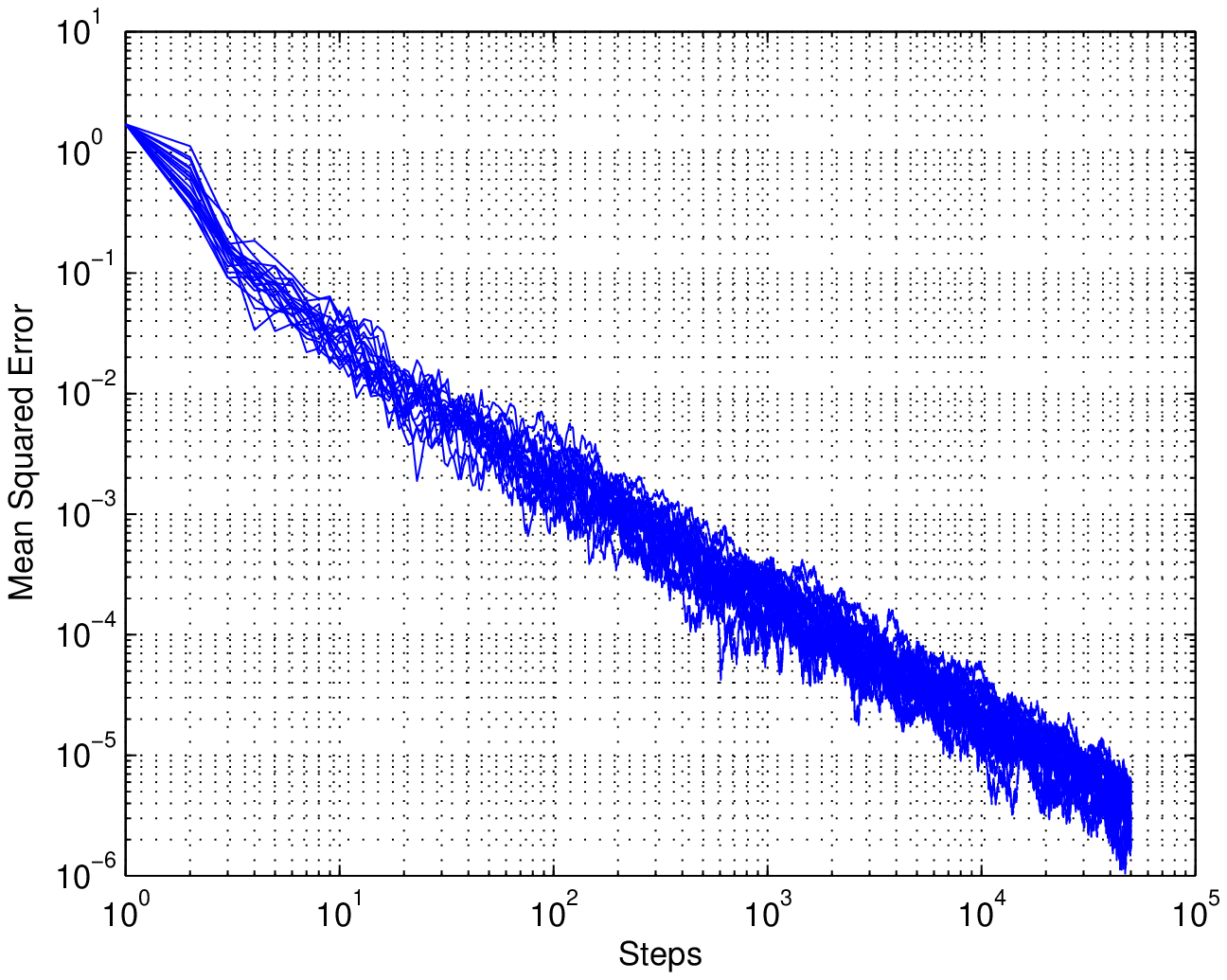} &
\widgraph{\figwtwo}{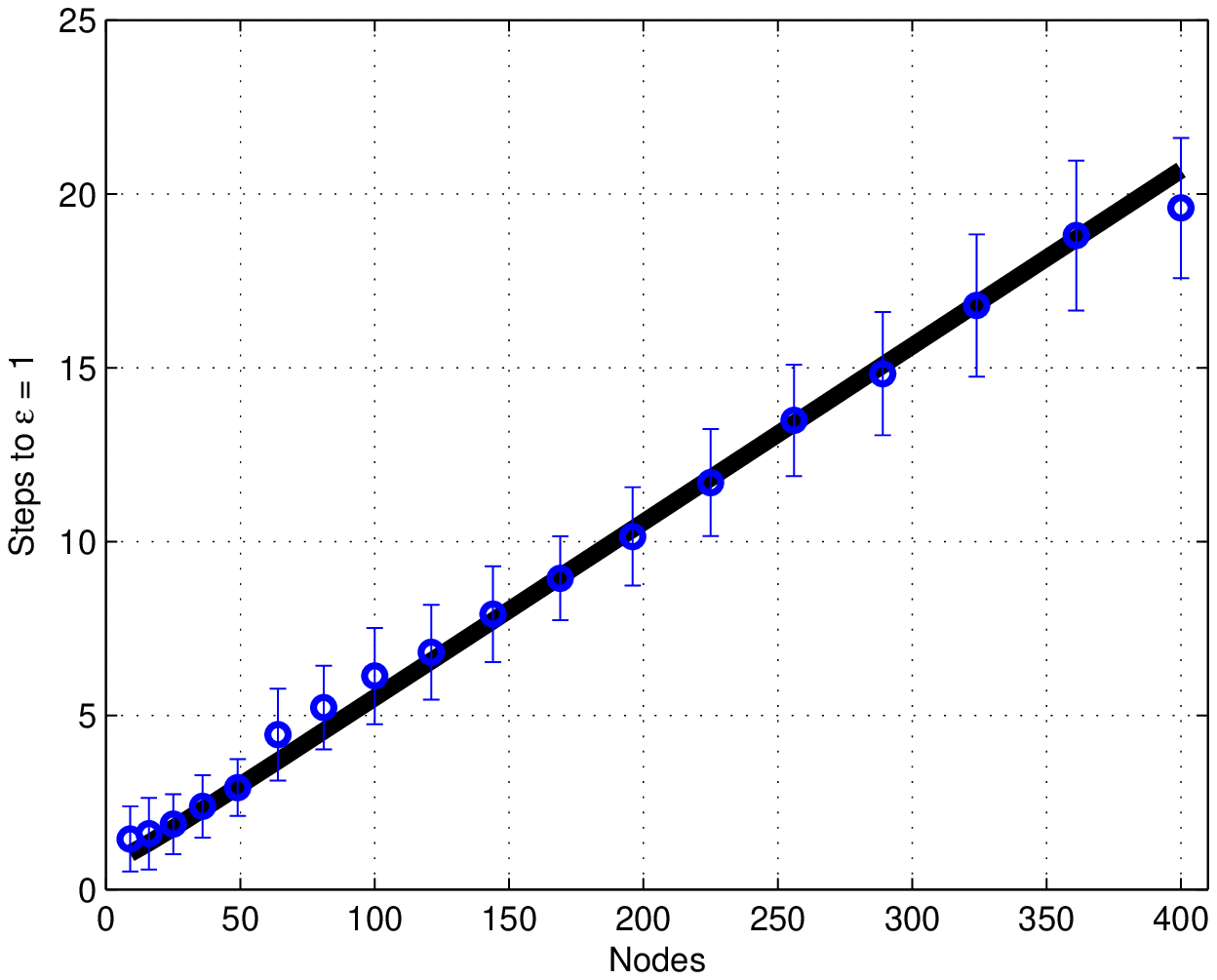} \\
(a) & (b) & (c)
\end{tabular}
\caption{Comparison of empirical simulations to theoretical
predictions for the four nearest-neighbor lattice (panel (a)).  (b)
Sample path plots of log MSE versus log iteration number: as predicted
by the theory, the log MSE scales linearly with log iterations.  (c)
Plot of number of iterations (vertical axis) required to reach a fixed
level of MSE versus the graph size (horizontal axis).  For the
lattice, graph, this quantity scales linearly in the graph size,
consistent with Corollary~\ref{CorAMSE}.}
\label{FigGrid}
\end{center}
\end{figure}

\subsubsection{Expander graphs}

Consider a bipartite graph $\graph = (V_1, V_2, \edge)$, with
$\numnode = |V_1| + |V_2|$ vertices and edges joining only vertices in
$V_1$ to those in $V_2$, and constant degree $\degmax$; see
Figure~\ref{FigExpander}(a) for an illustration with $\degmax = 3$.  A
bipartite graph of this form is an expander~\cite{Alo86,Alon,Chung}
with parameters $\alpha, \delta \in (0,1)$, if for all subsets $S
\subset V_1$ of size $|S| \leq \alpha |V_1|$, the neighborhood set of
$S$---namely, the subset
\begin{eqnarray*}
N(S) & \defn & \{t \in V_2 \; \mid \: (s,t) \quad \mbox{for some $s \in S$}
\},
\end{eqnarray*}
has cardinality $|N(S)| \geq \delta \degmax |S|$.  Intuitively, this
property guarantees that each subset of $V_1$, up to some critical
size, ``expands'' to a relatively large number of neighbors in $V_2$.
(Note that the maximum size of $|N(S)|$ is $\degmax |S|$, so that
$\delta$ close to $1$ guarantees that the neighborhood size is close
to its maximum, for all possible subsets $S$.)  Expander graphs have a
number of interesting theoretical properties, including the property
that $\eigval_2(\Lapgraph{\Fully_\numnode}) = \Theta(1)$---that is, a
bounded spectral gap~\cite{Alo86,Chung}.

In order to investigate the behavior of our algorithm for expanders,
we construct a random bi-partite graph as follows: for an even number
of nodes $\numnode$, we split them into two subsets $V_i, i=1,2$, each
of size $\numnode/2$.  We then fix a degree $\degmax$, construct a
random matching on $\degmax \frac{\numnode}{2}$ nodes, and use it
connect the vertices in $V_1$ to those in $V_2$.  This procedure forms
a random bipartite $\degmax$-regular graph; using the probabilistic
method, it can be shown to be an edge-expander with with probability
$1-o(1)$, as the graph size tends to infinity~\cite{Alo86,Feldman05b}.

\begin{figure}[h]
\begin{center}
\begin{tabular}{ccc}
\raisebox{.05in}{\widgraph{.10\textwidth}{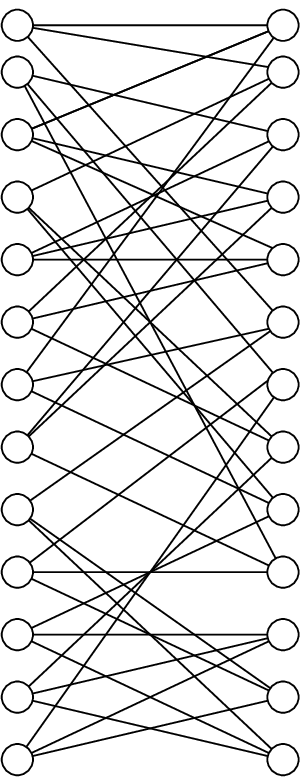}} & 
\widgraph{\figwtwo}{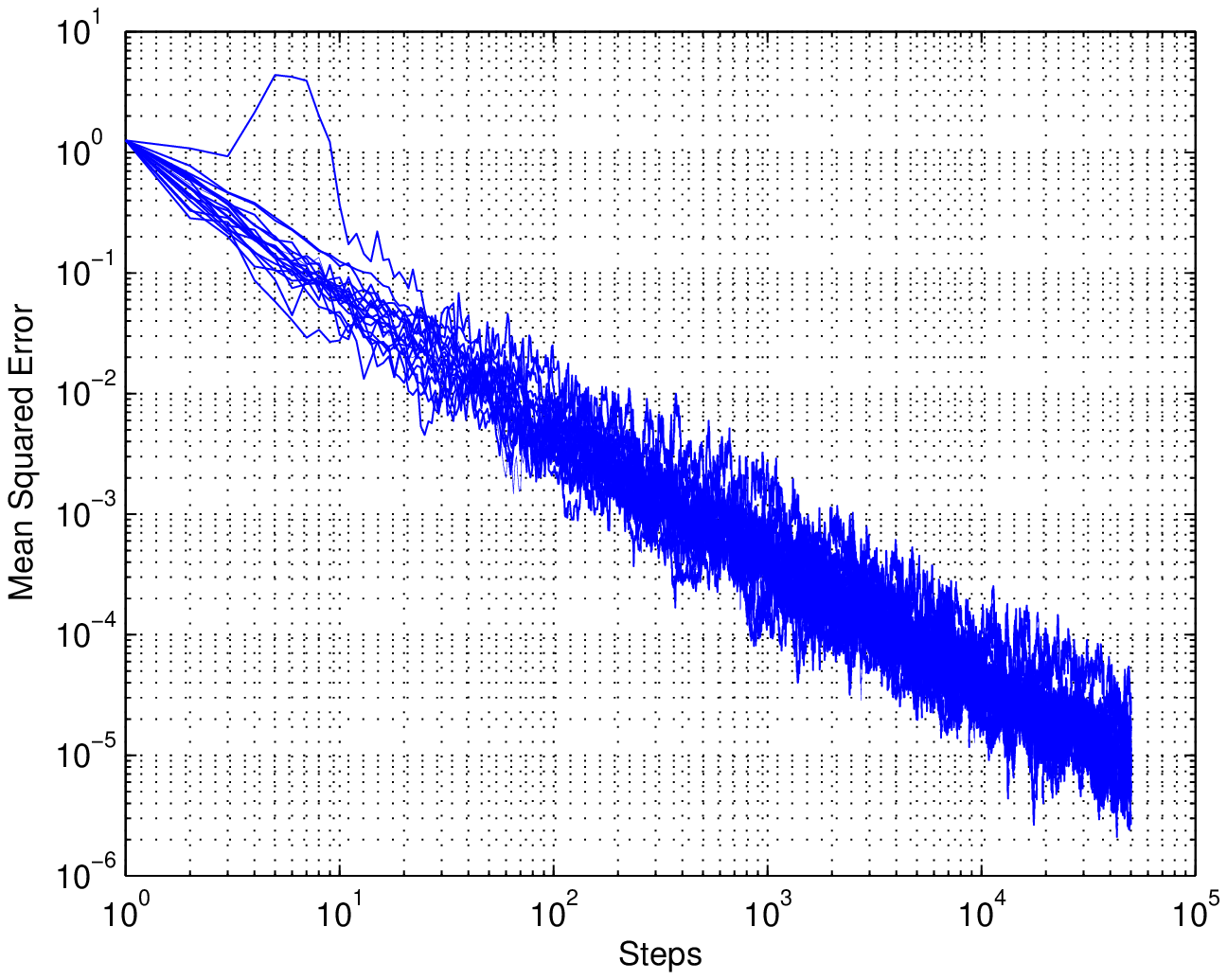} &
\widgraph{\figwtwo}{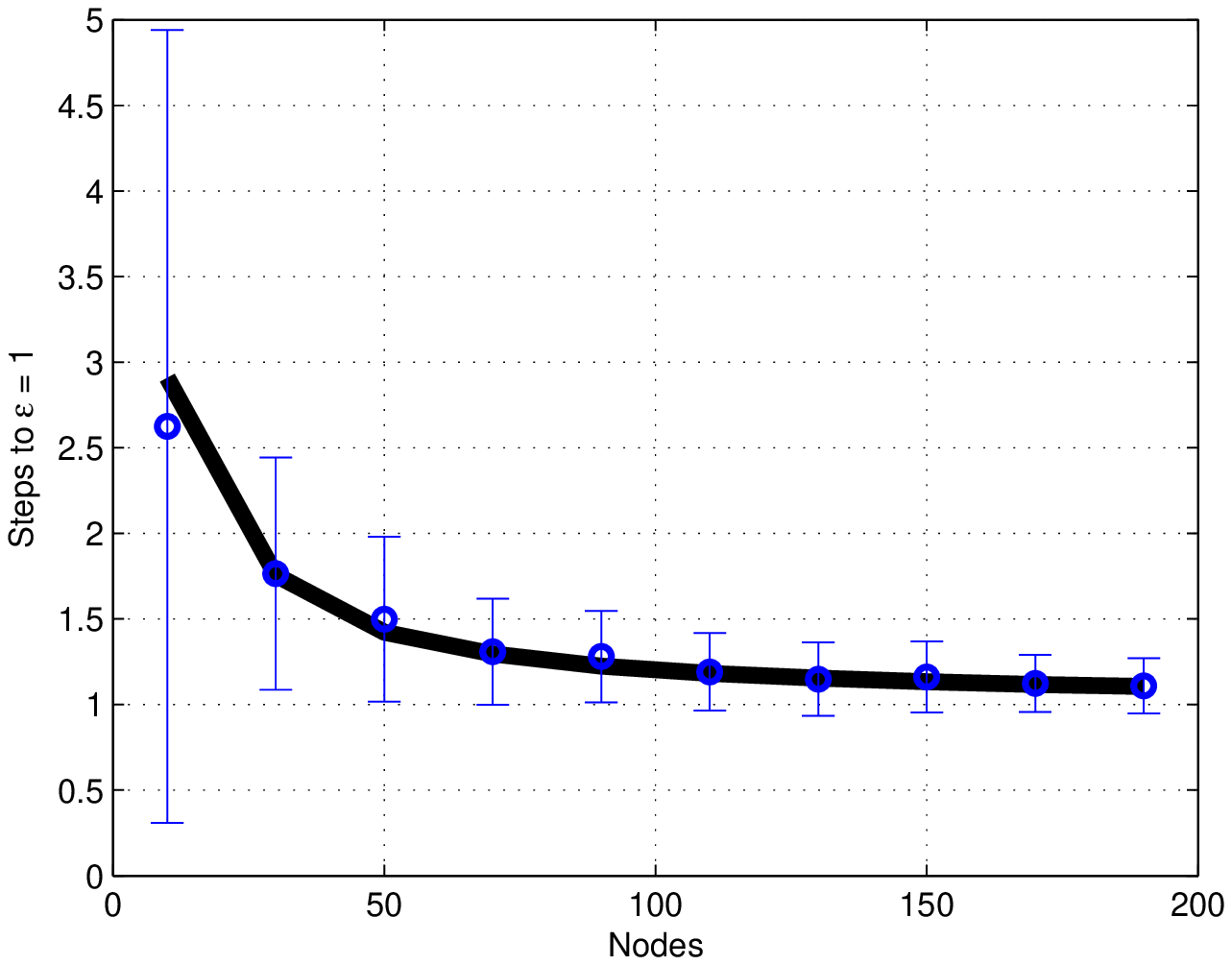} \\
(a) & (b) & (c)
\end{tabular}
\caption{Comparison of empirical simulations to theoretical
predictions for the bipartite expander graph in panel (a).  (b) Sample
path plots of log MSE versus log iteration number: as predicted by the
theory, the log MSE scales linearly with log iterations.  (c) Plot of
number of iterations (vertical axis) required to reach a fixed level
of MSE versus the graph size (horizontal axis).  For an expander, this
quantity remains essentially constant with the graph size, consistent
with Corollary~\ref{CorAMSE}.}
\label{FigExpander}
\end{center}
\end{figure}

Given the constant spectral gap $\eigval_2(\Lapgraph{\Fully_\numnode})
= \Theta(1)$, the scaling in number of iterations to achieve constant
MSE is $\itnum = \Theta(1)$.  This theoretical prediction is compared
to simulation results in Figure~\ref{FigExpander}; note how the number
of iterations soon settles down to a constant, as predicted by the
theory.

\section{Proof of Theorem~\ref{ThmConsist}}
\label{SecProof}

We now turn to the proof of Theorem~\ref{ThmConsist}.  The basic idea
is to relate the behavior of the stochastic
recursion~\eqref{EqnCompactUpdate} to an ordinary differential
equation (ODE), and then use the ODE method~\cite{Ljung77} to analyze
its properties.  The ODE involves a function $t \mapsto \pestcts{t}
\in \real^\numnode$, with its specific structure depending on the
communication and noise model under consideration.  For the $\aen$ and
$\ann$ models, the relevant ODE is given by
\begin{eqnarray}
\label{EqnODESimple}
\frac{d \pestcts{t}}{dt} & = & -\Lmat \pestcts{t}.
\end{eqnarray}
For the $\bcm$ model, the approximating ODE is given by
\begin{equation}
\label{EqnODESat}
\frac{d \pestcts{t}}{dt}  =  -\Lmat \; \satfun{\pestcts{t}} \qquad \mbox{with }
\satfun{u}  \, \defn \,  \begin{cases}  u & \mbox{if $|u| < M$} \\
-M & \mbox{if $u \leq -M$} \\
+M & \mbox{if $u \geq +M$.}
              \end{cases}
\end{equation}
In both cases, the ODE must satisfy the initial condition
$\pestcts{0}(\vertex) = \X(\vertex)$.

\subsection{Proof of Theorem~\ref{ThmConsist}(a)}

The following result connects the discrete-time stochastic process
$\{\pest{\itnum} \}$ to the deterministic ODE solution, and establishes
Theorem~\ref{ThmConsist}(a):
\begin{lemma}
\label{LemODEs}
The ODEs~\eqref{EqnODESimple} and~\eqref{EqnODESat} each have $\pstar
= \Xbar \ones$ as their unique stable fixed point.  Moreover, for all
$\delta > 0$, we have
\begin{eqnarray}
\label{EqnStrongPath}
\prob\left(\limsup_{n\rightarrow \infty}\|\theta_n - \theta_{t_n}\|>
\delta\right) & = & 0, \qquad \mbox{for $t_n = \sum_{k = 1}^{n}
\frac{1}{k}$},
\end{eqnarray}
which implies that $\pest{\itnum} \rightarrow \pstar$ almost surely.
\end{lemma}
\begin{proof}
We prove this lemma by using the ODE method and stochastic
approximation---in particular, Theorem 1 from Kushner and
Yin~\cite{Kushner97}, which connects stochastic recursions of the
form~\eqref{EqnCompactUpdate} to the ordinary differential equation $d
\pestcts{t}/dt = \expec_{\Xivar{}}[\itnum \, \left(\Lmat \odot
\Uvar{}(\pestcts{t}, \Xivar{}) \right) \mid \, \pestcts{t}]$.  Using
the definition of $\Uvar{}$ in terms of $\Gfun{}$, for the $\aen$ and
$\ann$ models, we have
\begin{eqnarray*}
  \expec_{\Xivar{}} \left[\Gfun{}(\pest{}(\vertex), \Xivar{}(\vertex,
\rec)) \, \mid \, \pest{}(\vertex) \right] & = & \pest{}(\vertex),
\end{eqnarray*}
from which we conclude that with the stepsize choice
$\stepsize_\numnode = \Theta(1/\numnode)$, we have
\begin{eqnarray*}
\expec_{\Xivar{}}[\itnum \, \left(\Lmat \odot
\Uvar{}(\pestcts{t}, \Xivar{}) \right) \mid \, \pestcts{t}] & = & -\Lmat
 \pestcts{t}.
\end{eqnarray*}
By our assumptions on the eigenstructure of $\Lmat$, the system $d
\pestcts{t}/dt = -\Lmat \, \pestcts{t}$ is globally asymptotically
stable, with a line of fixed points $\{ \pestcts{} \in \real^\numnode
\, \mid \,\Lmat \pestcts{} = 0 \}$.  Given the initial condition
$\pestcts{0}(\vertex) = \X(\vertex)$, we conclude that $\pstar = \Xbar
\ones$ is the unique asymptotically fixed point of the ODE, so that
the claim~\eqref{EqnStrongPath} follows from Kushner and
Yin~\cite{Kushner97}.

For the $\bcm$ model, the analysis is somewhat more involved, since the
quantization function saturates the output at $\pm M$.  For the dithered
quantization model~\eqref{EqnBCM}, we have
\begin{eqnarray*}
\expec_{\Xivar{}}[\itnum \, \left(\Lmat \odot \Uvar{}(\pestcts{t},
\Xivar{}) \right) \mid \, \pestcts{t}] & = & -\Lmat \,
\satfun{\pestcts{t}},
\end{eqnarray*}
where $\satfun{\cdot}$ is the saturation function~\eqref{EqnODESat}.
We now claim that $\pstar$ is also the unique asymptotically stable
fixed point of the ODE $d \pestcts{t}/dt = -\Lmat \,
\satfun{\pestcts{t}}$ subject to the initial condition
$\pestcts{0}(\vertex) = \X(\vertex)$.  Consider the eigendecomposition
$\Lmat = \Vmat \Dlmat \Vmat^T$, where $\Dlmat = \diag \{0,
\eigval_2(\Lmat), \ldots, \eigval_\numnode(\Lmat) \}$.  Define the
rotated variable $\gammacts{t} \defn \Vmat^T \pestcts{t}$, so that
the ODE~\eqref{EqnODESat} can be re-written as
\begin{subequations}
\begin{eqnarray}
\label{EqnGammaOne}
d \gammacts{t}(1)/dt & = & 0 \\
\label{EqnGammaRest}
d \gammacts{t}(k)/dt & = & -\eigval_{k}(\Lmat) \Vmat_k^T \satfun{\Vmat
  \gammacts{t}}, \qquad \mbox{for $k=2, \ldots, \numnode$},
\end{eqnarray}
\end{subequations}
where $\Vmat_k$ denotes the $k^{th}$ column of $\Vmat$.

Note that $\Vmat_1 = \ones/\|\ones\|_2$, since it is associated with
the eigenvalue $\eigval_1(\Lmat) = 0$.  Consequently, the solution to
equation~\eqref{EqnGammaOne} takes the form
\begin{equation}
\label{EqnGammaOneNew}
\gammacts{t}(1) = \Vmat_1^T \pestcts{0} \; = \; \sqrt{\numnode} \,
\Xbar,
\end{equation}
with unique fixed point $\gamstar(1) = \sqrt{\numnode} \, \Xbar$,
where $\Xbar \defn \frac{1}{\numnode} \sum_{i=1}^\numnode \X(i)$ is
the average value,

A fixed point $\gamstar \in \real^\numnode$ for
equations~\eqref{EqnGammaRest} requires that $\Vmat_k^T \satfun{\Vmat
\gamstar} = 0$, for $k = 2, \ldots, \numnode$.  Given that the columns
of $\Vmat$ form an orthogonal basis, this implies that $\satfun{\Vmat
\gamstar} = \alpha \ones$ for some constant $\alpha \in \real$, or equivalently
(given the connection $\Vmat \gamstar = \pstar$)
\begin{equation}
\satfun{\pstar} \; = \; \alpha \ones.
\end{equation}
Given the piecewise linear nature of the saturation function, this
equality implies either that the fixed point satisfies the elementwise
inequality $\pstar > M$ (if $\alpha = M$); or the elementwise
inequality $\pstar < -M$ (if $\alpha = -M$); or as the final option,
the $\pstar = \alpha$ when $\alpha \in (-M, +M)$.  But from
equation~\eqref{EqnGammaOneNew}, we know that $\gamstar(1) =
\sqrt{\numnode} \, \Xbar \in [-M \, \sqrt{\numnode}, \; +M
\sqrt{\numnode}]$.  But we also have $\gamstar(1) =
\frac{\ones^T}{\sqrt{\numnode}} \, \pstar$ by definition, so that
putting together the pieces yields
\begin{equation}
-M \; < \; \frac{\ones \, \pstar}{\nodenum} \; < \; M,
\end{equation}
Thus the only possibility is that $\pstar = \alpha \ones$ for some
constant $\alpha \in (-M, +M)$, and the relation $\Vmat \gamstar =
\alpha \ones$ implies that $\alpha = \gamstar(1)/\sqrt{\nodenum} =
\Xbar$, which establishes the claim.
\end{proof}

\subsection{Proof of Theorem~\ref{ThmConsist}(b)}

We analyze the update~\eqref{EqnCompactUpdate} using results from
Benveniste et al~\cite{Benveniste90}.  In particular, given the
stochastic iteration $\pest{\itnum+1} = \pest{\itnum} +
\stepsize_\itnum H(\pest{\itnum},\Uvar{\itnum+1})$, define the
expectation $h(\pest{}) = \expec\left[H(\pest{},X)\right]$, its
Jacobian matrix $\nabla h(\pest{})$, and the covariance matrix
$\Condcov{}(\pest{}) = \expec\left[(H(\pest{},X) -
h(\pest{}))(H(\pest{},X) - h(\pest{}))^T\right]$.  Then Theorem 3
(p. 110) of Benveniste et al~\cite{Benveniste90} asserts that as long
as the eigenvalues $\eigval(\nabla h(\pest{}))$ are strictly below
$-1/2$. then
\begin{equation}
\label{eq:asnorm}
\sqrt{n}\left(\theta_n-\theta^*\right) \convdist N(0,Q),
\end{equation}
where the covariance matrix $Q$ is the unique solution to the Lyapunov
equation
\begin{equation}
\left(\frac{I}{2}+ \nabla h(\pstar)\right) Q + Q\left(\frac{I}{2}+
\nabla h(\pstar)\right)^T + \Condcov{}_{\pstar}=0.
\end{equation}

We begin by computing the conditional distribution $h(\pest{})$; for
the models $\aen$ and $\ann$ it takes the form
\begin{eqnarray}
h(\pest{}) & = & - \Lmat \; \pest{}
\end{eqnarray}
since the conditional expectation of the random matrix $\Uvar{}$ is
given by $\Exs[\Uvar{} \, \mid \, \pest{}] = \pest{} \; \ones$.  For
the $\bcm$ model, since the quantization is finite with maximum value
$M$, the expectation is given by
\begin{align}
h(\pest{}) &= -\Lmat \; \satfun{\pest{}}
\end{align}
where the saturation function was defined
previously~\eqref{EqnODESat}.  In addition, we computed form of the
the covariance matrix $\Condcov{}_{\pstar}$
previously~\eqref{EqnDefnCondcov}.  Finally, we note that
\begin{eqnarray}
\nabla h(\pstar) & = & -\Lmat
\end{eqnarray}
for all three models.  (This fact is immediate for models $\aen$ and
$\ann$; for the $\bcm$ model, note that Theorem~\ref{ThmConsist}(a)
guarantees that $\pstar$ falls in the middle linear portion of the
saturation function.)

\newcommand{\gammaest}[1]{\ensuremath{\beta^{#1}}}
\newcommand{\gammaeststar}{\ensuremath{\beta^*}}

We cannot immediately conclude that asymptotic
normality~\eqref{eq:asnorm} holds, because the matrix $\Lmat$ has a
zero eigenvalue ($\eigval_1(\Lmat) = 0$).  However, let us decompose
$\Lmat = \Vmat \Dlmat \Vmat^T$ where $\Vmat$ is the matrix with unit
norm columns as eigenvectors, and $\Dlmat = \diag \{0,
\eigval_2(\Lmat), \ldots, \eigval_\numnode(\Lmat) \}$.  Let $\Vtil$
denote the $\numnode \times (\numnode-1)$ matrix obtained by deleting
the first column of $\Vmat$.  Defining the $(\numnode-1)$ vector
$\gammaest{\itnum} = \Vtil^T \pest{\itnum}$, we can rewrite the update
in $(\numnode-1)$-dimensional space as
\begin{eqnarray}
\gammaest{\itnum+1} & = & \gammaest{\itnum}+\frac{1}{\itnum}
\left[-\Vmat^T\left(L\,\odot
\,\Uvar{\itnum+1}(\pest{\itnum})\right)\,\ones \right],
\end{eqnarray}
for which the new effective $h$ function is given by
$\tilde{h}(\gammaest{}) = -\Dlmattil \gammaest{}$, with $\Dlmattil =
\diag \{ \eigval_2(\Lmat), \ldots, \eigval_\numnode(\Lmat) \}$.  Since
$\eigval_2(\Lmat) > \frac{1}{2}$ by assumption, the asymptotic
normality~\eqref{eq:asnorm} applies to this reduced iteration, so that
we can conclude that
\begin{eqnarray*}
\sqrt{\itnum} \,  (\gammaest{\itnum} - \gammaeststar) & \convdist &  N(0, \Qmattil)
\end{eqnarray*}
where $\Qmattil$ solves the Lyapunov equation
\begin{equation*}
\left(\Dlmattil - \frac{I}{2} \right) \Qmattil + \Qmattil
\left(\Dlmattil - \frac{I}{2} \right)^T \; = \; \Vtil^T
\Condcov{}_{\pstar} \Vtil.
\end{equation*}
We conclude by noting that the asymptotic covariance of $\pest{\itnum}$
is related to that of $\gammaest{\itnum}$ by the relation
\begin{eqnarray}
\Qmat & = & \Vmat^T \begin{bmatrix} 0 & 0 \\ 0 & \Qmattil
\end{bmatrix} \Vmat,
\end{eqnarray}
from which Theorem~\ref{ThmConsist}(b) follows.

%%%%%%%%%%%%%%%%%%%%%%%%%%%%%%%%%%%%%%%%%%%%%%%%%%%%%%%%%%%%%%%%%%%%%%%%%%%%%%%%%%%%%%%%%%%%%%%%%%%%%%%%%%%%%%%%%%%%%%%%%%%%%%%%%%%%%%%%%%%%%%

\section{Discussion}
\label{SecDiscuss}

This paper analyzed the convergence and asymptotic behavior of
distributed averaging algorithms on graphs with general noise models.
Using suitably damped updates, we showed that it is possible to obtain
exact consensus, as opposed to approximate or near consensus, even in
the presence of noise.  We guaranteed almost sure convergence of our
algorithms under fairly general conditions, and moreover, under
suitable stability conditions, we showed that the error is
asymptotically normal, with a covariance matrix that can be predicted
from the structure of the consensus operator.  We provided a number of
simulations that illustrate the sharpness of these theoretical
predictions.  Although the current paper has focused exclusively on
the averaging problem, the methods of analysis in this paper are
applicable to other types of distributed inference problems, such as
computing quantiles or order statistics, as well as computing various
types of $M$-estimators.  Obtaining analogous results for more general
problems of distributed statistical inference is an interesting
direction for future research.

\ifthenelse{\equal{\doctype}{TECH}}
{
\subsection*{Acknowledgements}
This work was presented in part at the Allerton Conference on Control,
Computing and Communication, September 2007.  Work funded by
NSF-grants DMS-0605165 and CCF-0545862 CAREER to MJW.  The authors
thank Pravin Varaiya and Alan Willsky for helpful comments.

\bibliographystyle{plain} 
}
%%%%%%%%%%%%%%%%%%%%%%%%%%%%%%%%%%%%%%%%%%%%%%%%%%%%%%%%%%%%%%%%%%%%%%%%
{ \thanks{This work was presented in part
at the Allerton Conference on Control, Computing and Communication,
September 2007.  Work funded by NSF-grants DMS-0605165, CCF-0545862
CAREER and a Sloan Foundation Fellowship to MJW.}

\bibliographystyle{plain}
}

%%%%%%%%%%%%%%%%%%%%%%%%%%%%%%%%%%%%%%%%%%%%%%%%%%%%%%%%%%%%%%%%%%%%%
\appendix

%%%%%%%%%%%%%%%%%%%%%%%%%%%%%%%%%%%%%%%%%%%%%%%%%%%%%%%%%%%%%%%%%%%

\section{Proof of Lemma~\ref{LemOrder}}
\label{AppLemOrder}
We begin by noting that for the normalized graph Laplacian
$\Lapgraph{\graph}$, it is known that for any graph, the second
smallest eigenvalue satisfies the upper bound
$\eigval_{2}(\Lapgraph{\graph}) \leq \numnode/(\numnode-1) \leq 1$.
Moreover, we have $\trace(\Lapgraph{\graph}) = \numnode$.  See Lemma
1.7 in Chung~\cite{Chung} for proofs of these claims.

Using these facts, we establish Lemma~\ref{LemOrder} as follows.
Recall that by construction, we have $\Lstar =
\frac{\Lapgraph{\graph}}{\eigval_2(\Lapgraph{\graph})}$, so that the
second smallest eigenvalue of $\Lstar$ is $\eigval_2(\Lstar) = 1$, and
the remaining eigenvalues are greater than or equal to one.  Applying
Corollary~\ref{CorAMSE} to the $\ann$ model, we have
\begin{eqnarray*}
\amse(\Lmat; \pstar) & = & \frac{\sigma^2}{\numnode} \; \;
\sum_{i=2}^{\numnode} \; \left[\frac{[\eigval_i(\Lstar)]^2}{2
\eigval_i(\Lstar) - 1} \right],\\
& = & \frac{\sigma^2}{\numnode \; \eigval_2(\Lapgraph{\graph})} \; \;
\sum_{i=2}^{\numnode} \left[\frac{[\eigval_i(\Lapgraph{\graph})]^2}{2
\eigval_i(\Lapgraph{\graph}) - \eigval_2(\Lapgraph{\graph})} \right] \\
& \geq & \frac{\sigma^2}{2 \eigval_2(\Lapgraph{\graph}) \; \numnode}
\; \trace(\Lapgraph{\graph}) \\
& = & \frac{\sigma^2}{2 \eigval_2(\Lapgraph{\graph})}
\end{eqnarray*}
using the fact that $\trace(\Lapgraph{\graph}) = \numnode$.

In the other direction, using the fact that $\eigval_2(\Lstar) \geq 1$
and the bound $\frac{x^2}{2 x - 1} \leq x$ for $x \geq 1$. we have
\begin{eqnarray*}
\amse(\Lmat; \pstar) & = & \frac{\sigma^2}{\numnode} \; \;
\sum_{i=2}^{\numnode} \; \left[\frac{[\eigval_i(\Lstar)]^2}{2
\eigval_i(\Lstar) - 1} \right],\\
& \leq & \frac{\sigma^2}{\numnode} \trace(\Lstar) \\
& = & \frac{\sigma^2}{\eigval_{2}(\Lapgraph{\graph})\,\numnode}
\trace(\Lapgraph{\graph}) \\
& = & \frac{\sigma^2}{\eigval_2(\Lapgraph{\graph})}.
\end{eqnarray*}

%%%%%%%%%%%%%%%%%%%%%%%%%%%%%%%%%%%%%%%%%%%%%%%%%%%%%%%%%%%%%%%%%%

%\bibliography{MERGED_ram_mjwain}

\begin{thebibliography}{10}

\bibitem{Alo86}
N.~Alon.
\newblock Eigenvalues and expanders.
\newblock {\em Combinatorica}, 6(2):83--96, 1986.

\bibitem{Alon}
N.~Alon and J.~Spencer.
\newblock {\em The {P}robabilistic {M}ethod}.
\newblock Wiley Interscience, New York, 2000.

\bibitem{AysCoaRab07}
T.~C. Aysal, M.~Coates, and M.~Rabbat.
\newblock Distributed average consensus using probabilistic quantization.
\newblock In {\em IEEE Workshop on Stat. Sig. Proc.}, Madison, WI, August 2007.

\bibitem{Benveniste90}
A.~Benveniste, M.~Metivier, and P.~Priouret.
\newblock {\em Adaptive Algorithms and Stochastic Approximations}.
\newblock Springer-Verlag, New York, NY, 1990.

\bibitem{BorVar82}
V.~Borkar and P.~Varaiya.
\newblock Asymptotic agreement in distributed estimation.
\newblock {\em IEEE Trans. Auto. Control}, 27(3):650--655, 1982.

\bibitem{Boyd06}
S.~Boyd, A.~Ghosh, B.~Prabhakar, and D.~Shah.
\newblock Randomized gossip algorithms.
\newblock {\em IEEE Transactions on Information Theory}, 52(6):2508--2530,
  2006.

\bibitem{Chung}
F.R.K. Chung.
\newblock {\em Spectral Graph Theory}.
\newblock American Mathematical Society, Providence, RI, 1991.

\bibitem{deGroot74}
M.~H. de{G}root.
\newblock Reaching a consensus.
\newblock {\em Journal of the American Statistical Association},
  69(345):118--121, March 1974.

\bibitem{DimSarWai08}
A.~G. Dimakis, A.~Sarwate, and M.~J. Wainwright.
\newblock Geographic gossip: {E}fficient averaging for sensor networks.
\newblock {\em IEEE Trans. Signal Processing}, 53:1205--1216, March 2008.

\bibitem{FagZam07}
F.~Fagnani and S.~Zampieri.
\newblock Average consensus with packet drop communication.
\newblock {\em {SIAM} J. on Control and Optimization}, 2007.
\newblock To appear.

\bibitem{Feldman05b}
J.~Feldman, T.~Malkin, R.~A. Servedio, C.~Stein, and M.~J. Wainwright.
\newblock {LP} decoding corrects a constant fraction of errors.
\newblock {\em IEEE Trans. Information Theory}, 53(1):82--89, January 2007.

\bibitem{Kashyap07}
A.~Kashyap, T.~Basar, and R.~Srikant.
\newblock Quantized consensus.
\newblock {\em Automatica}, 43:1192--1203, 2007.

\bibitem{Kempe03}
D.~Kempe, A.~Dobra, and J.~Gehrke.
\newblock Gossip-based computation of aggregate information.
\newblock {\em Proc. 44th Ann. IEEE FOCS}, pages 482--491, 2003.

\bibitem{Kushner97}
H.~J. Kushner and G.~G. Yin.
\newblock {\em Stochastic Approximation Algorithms and Applications}.
\newblock Springer-Verlag, New York, NY, 1997.

\bibitem{Ljung77}
L.~Ljung.
\newblock Analysis of recursive stochastic algorithms.
\newblock {\em IEEE Transactions in Automatic Control}, 22:551--575, 1977.

\bibitem{Olfati07}
R.~Olfati-Saber, J.~A. Fax, and R.~M. Murray.
\newblock Consensus and cooperation in networked multi-agent systems.
\newblock {\em Proceedings of the IEEE}, 95(1):215--233, 2007.

\bibitem{Schizas08}
I.~D. Schizas, A.~Ribeiro, and G.~B. Giannakis.
\newblock Consensus in ad hoc {WSN}s with noisy links: Part {I} distributed
  estimation of deterministic signals.
\newblock {\em IEEE Transactions on Signal Processing}, 56(1):350--364, 2008.

\bibitem{Tsitsiklis84}
J.~Tsitsiklis.
\newblock {\em Problems in decentralized decision-making and computation}.
\newblock PhD thesis, Department of EECS, MIT, 1984.

\bibitem{Xiao04}
L.~Xiao and S.~Boyd.
\newblock Fast linear iterations for distributed averaging.
\newblock {\em Systems \& Control Letters}, 52:65--78, 2004.

\bibitem{Xiao07}
L.~Xiao, S.~Boyd, and S.-J. Kim.
\newblock Distributed average consensus with least-mean-square deviation.
\newblock {\em Journal of Parallel and Distributed Computing}, 67(1):33--46,
  2007.

\bibitem{YilSca07}
M.~E. Yildiz and A.~Scaglione.
\newblock Differential nested lattice encoding for consensus problems.
\newblock In {\em Info. Proc. Sensor Networks (IPSN)}, Cambridge, MA, April
  2007.

\end{thebibliography}

%%%%%%%%%%%%%%%%%%%%%%%%%%%%%%%%%%%%%%%%%%%%%%%%%%%%%%%%%%%%%%%%%%%%%%%%%%

\end{document}